\documentclass{llncs}
\usepackage{amsmath}
\usepackage{amssymb}
\usepackage{bm}
\usepackage{enumitem}
\usepackage{framed}

\newcommand{\Oh}{\mathcal{O}}

\newcommand\textline[5][t]{
  \par\smallskip\noindent\parbox[#1]{.61\textwidth}{\raggedright\textsc{#2}}
  \parbox[#1]{.12\textwidth}{\centering #3}
  \parbox[#1]{.12\textwidth}{\centering #4}
  \parbox[#1]{.14\textwidth}{\raggedleft #5}\par\smallskip
}
\newcommand\textlineii[5][t]{
  \par\smallskip\noindent\parbox[#1]{.24\textwidth}{\raggedright\textsc{#2}}
  \parbox[#1]{.20\textwidth}{\centering #3}
  \parbox[#1]{.30\textwidth}{\centering #4}
  \parbox[#1]{.25\textwidth}{\raggedleft #5}\par\smallskip
}

\newcommand\textlineiii[5][t]{
  \par\smallskip\noindent\parbox[#1]{.46\textwidth}{\raggedright\textsc{#2}}
  \parbox[#1]{.27\textwidth}{\centering #3}
  \parbox[#1]{.12\textwidth}{\centering #4}
  \parbox[#1]{.14\textwidth}{\raggedleft #5}\par\smallskip
}

\newcommand\problemline[4]{\textline[t]{#1}{\proj{#2}}{$\Oh($#3$)$}{$\Theta($#4$)$}}
\newcommand\problemlinealt[4]{\textline[t]{#1}{#2}{#3}{#4}}
\newcommand\problemlinealtt[4]{\textlineiii[t]{#1}{#2}{#3}{#4}}

\newcommand\rest{\hspace{-1mm}\upharpoonright}
\newcommand\proj[1]{$\textrm{proj}_{#1}$}
\newcommand\ex{$\textrm{ex}$}

\def\yy{\bm{y}}
\def\zz{\bm{z}}
\def\gg{\bm{g}}
\def\ff{\bm{f}}
\def\BB{\mathcal{B}}
\def\CC{\mathcal{C}}
\def\DD{\mathcal{D}}
\def\FF{\mathcal{F}}
\def\HH{\mathcal{H}}

\def\KK{\mathcal{K}}
\def\CK{\mathcal{K}}
\def\RP{\mathcal{R}}
\def\NN{\mathbb{N}}
\def\RR{\mathbb{R}}

\begin{document}

\title{Extended Formulation for CSP that is Compact\\
for Instances of Bounded Treewidth\thanks{This research 
was partially supported by the project 14-10003S of GA \v{C}R.}}
\author{Petr Kolman, Martin Kouteck\' y}
\institute{Department of Applied Mathematics, Faculty
of Mathematics and Physics,\\ Charles University in Prague, Czech Republic\\
\email{kolman@kam.mff.cuni.cz, koutecky@kam.mff.cuni.cz} }

\maketitle

\begin{abstract}
In this paper we provide an extended formulation for the class
of constraint satisfaction problems and prove that its size is 
polynomial for instances whose constraint graph has bounded treewidth.
This implies new upper bounds on extension complexity of several
important NP-hard problems on graphs of bounded treewidth.
\end{abstract}
\section{Introduction}

Many important combinatorial optimization problems belong to the class
of constraint satisfaction problems (CSP). 
Naturally, a lot of effort has been given to
design efficient approximation algorithms for CSP, to prove
complexity lower bounds for CSP, and to identify tractable instances
of CSP (e.g., from the point of view of parameterized complexity). It
has been shown that CSP is solvable in polynomial time for instances
whose constraint graph has bounded treewidth~\cite{Freuder:90}.

In recent years, a lot of attention has been given to study {\em
extension complexity} of problems~\cite{CCZ:13}: 
what is the minimum number of inequalities representing a
polytope whose (suitably chosen) linear projection coincides
with the convex hull $H$ of all integral solutions of $Q$?
Such a polytope is called the {\em extended formulation} of $H$.
Note that membership of a problem in the class P of polynomially
solvable problems does not necessarily imply the existence of an
extended formulation of polynomial size~\cite{Roth:14}. 
In this work, we present an extended formulation for CSP and show that
its size is polynomial for instances of CSP whose constraint graph 
has bounded treewidth.

\subsection{Notation and Terminology}
An instance $Q=(V,\DD,\HH,\CC)$ of CSP
consists of
\begin{itemize}
\item a set of {\em variables} $z_v$, one for each $v\in V$; without loss of 
generality we assume that $V = \{1,\ldots,n\}$,
\item a set $\DD$ of finite {\em domains} $D_v\subseteq \RR$ (also 
denoted $D(v)$), one for each $v\in V$,
\item a set of {\em hard constraints}
  $\HH \subseteq \{C_{U}\ |\ U \subseteq V \}$ where each
hard constraint $C_{U} \in \HH$ with $U=\{i_1, i_2,\dots,i_k\}$ and 
$i_1 < i_2 < \cdots < i_k$, is a $|U|$-ary relation
$C_U \subseteq D_{i_1}\times D_{i_2}\times \cdots \times D_{i_k}$,
\item a set of {\em soft constraints}
  $\CC \subseteq \{C_U \ |\ U \subseteq V\}$ where each soft constraint
$C_{U} \in \CC$ with $U=\{i_1, i_2,\dots,i_k\}$ and 
$i_1 < i_2 < \cdots < i_k$, is a $|U|$-ary relation
$C_{U}\subseteq D_{i_1}\times D_{i_2}\times \cdots \times D_{i_k}$.
\end{itemize}
The \emph{constraint graph} of $Q$ is defined as $G=(V,E)$
where $E= \{\{u,v\} \ |\  \exists C_{U} \in \CC\cup \HH \textrm{ s.t. } 
\{u,v\} \subseteq U\}$. We say that a {\em CSP instance $Q$ has bounded 
treewidth} if the constraint graph of $Q$ has bounded treewidth.
In {\em binary CSP}, every hard and soft relation is a unary or binary
relation, and in  {\em boolean CSP}, the domain of every variable is
$\{0,1\}$.
We use $D$ to denote the maximal size of all domains, 
that is, $D=\max_{u\in V}|D_u|$.

For a vector $z^{}=(z^{}_1, z_2, \ldots,z^{}_n)$ and $U=\{i_1,
i_2,\dots,i_k\}\subseteq V$ with $i_1 < i_2 < \cdots < i_k$, we define
the {\em projection of} $z$ on $U$ as 
$z^{}|_U=(z^{}_{i_1}, z_{i_2}, \ldots, z^{}_{i_k})$.
A vector $z\in \RR^n$ {\em satisfies the constraint} $C_U \in
\CC\cup \HH$ if and only if $z|_U \in C_U$. 
We say that a vector
$z^{\star}=(z^{\star}_1,\ldots,z^{\star}_n)$ is {\em a feasible assignment} for
$Q$ if $z^{\star} \in D_1\times D_2\times \ldots\times D_n$ and 
$z^{\star}$ satisfies every hard constraint $C\in \HH$.
For a given feasible assignment $z^{\star}$ we define an {\em extended
feasible assignment} $\ex(z^{\star})=(z^{\star},h^{\star})\in \RR^{n+|\CC|}$ as follows:
the coordinates of $h^{\star}$ are indexed by the soft constraints from $\CC$
(to be more precise: by the subsets $U$ of $V$ used as lower indices
of the soft constraints)
and for each $C_U\in \CC$, we have $h^{\star}_U=1$ if and only if $z^{\star}|_U\in C_U$, 
and $h^{\star}_U=0$ otherwise.
We denote by $\FF(Q)$ the set of all feasible assignments for $Q$, 
by $\FF^{ex}(Q)=\{\ex(z^{\star}) \ | \ z^{\star}\in \FF(Q)\}$ 
the set of all extended feasible assignments for~$Q$.
For every instance $Q$ we define two polytopes: 
$CSP(Q)$ is the convex hull of $\FF^{ex}(Q)$ and $CSP'(Q)$
the convex hull of $\FF(Q)$.
We also define
three trivial linear projections:
\begin{itemize}
\item $\textrm{proj}_V(z,h) = z$,
\ \ \ \ \ \ $\textrm{proj}_E(z,h) = h$,
\ \ \ \ \ \ $\textrm{proj}_{id}(z,h) = (z,h)$ 
\end{itemize}
where $z\in \RR^n$ and $h\in \RR^{|\CC|}$,
and observe that $\textrm{proj}_V(CSP(Q)) = CSP'(Q)$.

In the {\em decision} version of CSP, the set $\CC$ of soft constraints 
is empty and the task is to decide whether there
exists a feasible assignment. In
the {\em maximization} ({\em minimization}, resp.) version of the problem, the
task is to find a {feasible assignment} that maximizes (minimizes,
resp.) the number of \emph{satisfied} (unsatisfied, resp.) soft
constraints. Note that there is no difference between maximization and
minimization versions of the problem with respect to optimal solutions
but the two versions differ significantly from an approximation
perspective.

In the {\em weighted} version of CSP we are also given a weight
function $w: \CC \rightarrow \RR$ that specifies for each soft
constraint $C \in \CC$ its weight $w(C)$. The goal is to find a
feasible assignment that maximizes (minimizes, resp.) the total weight
of satisfied (unsatisfied, resp.)  constraints. The unweighted version
of CSP is equivalent to the weighted version with $w(C) = 1$
for all $C \in \CC$.

Even more generally, the relations in the soft constraints can be
replaced by bounded real valued payoff functions: 
a soft constraint $C_U \in \cal C$ with $U=\{i_1, i_2,\dots,i_k\}$ 
is not a $|U|$-ary relation but a function
$w:D_{i_1}\times D_{i_2}\times \ldots\times D_{i_k}\rightarrow \RR$
and the payoff of the soft constraint $C_U$ for a feasible assignment
$z^{\star}$ is $w(z^{\star}|_U)$; the objective
is to maximize (minimize, resp.) the total payoff. For the sake of
simplicity of the presentation we do not consider the problem in this
generality although the techniques used in this paper apply in the
general setting as well.

For notions related to the treewidth of a graph, we stick to the standard
terminology as given in the book by Kloks~\cite{Kloks:94}).

\subsection{Related Work}

\paragraph{CSP for graphs of bounded treewidth.}
As CSP captures many NP-hard problems, it is a natural problem to identify
tractable special cases of CSP. 
Freuder~\cite{Freuder:90} showed that
CSP instances with treewidth bounded by $\tau$ can be solved in time
$O(D^{\tau}n)$. Later, Grohe et al.~\cite{GSS:01} proved that, assuming $FPT\not=
W[1]$, this is essentially the only nontrivial class of graphs for which CSP is
solvable in polynomial time (cf.~Marx~\cite{Marx:10}). 

Describing the polytope of CSP solutions by the means of linear programming, 
for instances of bounded treewidth, is not a new idea. In 2007, Sellmann et
al. published a paper~\cite{SML:07} in which they described a linear program
that was supposed to define the convex hull of all feasible solutions of a
binary CSP when the constraint graph is a tree. They also provided a procedure
to convert a given CSP instance with bounded treewidth into one whose
constraint graph is a tree, at the cost of blowing up the number of
variables and constraints by a function of the treewidth.  
Unfortunately, there was a substantial bug in their proof and one of the main
theorems in the paper does not even hold~\cite{Sellmann:08}.

The paper~\cite{SML:07} also implicitely includes this folklore result: if the
constraint graph has treewidth at most $\tau$, then CSP can be solved by 
$\tau$ levels of the Sherali-Adams hierarchy. The resulting formulation is
of size $\Oh(n^\tau)$ while our approach yields size $\Oh(D^\tau n)$.


\paragraph{CSP for general graphs.}
%
Chan et al.~\cite{CLRS:13} study the extent to which linear programming
relaxation can be used in dealing with approximating CSP. They show
that polynomial-sized LPs are exactly as powerful as
LPs obtained from a constant number of rounds of the Sherali-Adams
hierarchy. They also prove integrality gaps for polynomial-sized LPs
for some CSP. 

Raghavendra~\cite{Ragha:08} shows that under the Unique Games Conjecture,
a certain simple SDP relaxation achieves the best approximation ratio
for every CSP.  In a follow up paper, Raghavendra and
Steurer~\cite{RS:09a} describe an efficient rounding scheme that
achieves the integrality gap of the simple SDP relaxation, and, in
another paper~\cite{RS:09b}, they show unconditionally that the
integrality gap of this SDP relaxation cannot be reduced by
Sherali-Adams hierarchies.

\paragraph{Other related results.} 
Buchanan and Butenko~\cite{BB:14} provide an extended formulation for the
independent set problem, a special case of CSP, that has size $O(2^{\tau}n)$
where $\tau$ denotes the treewidth of the given graph. Our results can be viewed
as a generalization of this result: the size of our formulation, when applied
to the independent set problem, is also $O(2^{\tau}n)$.

In a recent work, Bienstock and Munoz~\cite{BM:15} define a class of so called
{\em general binary} optimization problems which are essentially weighted boolean 
CSP problems,
and
for instances of treewidth $\tau$ provide an LP formulation of size $O(2^{\tau}n)$.
Again, this is a special case of our result in this paper.
It is worth mentioning at this point that every CSP instance can be transformed
into a boolean CSP instance; however, the standard transformation results in a
substantial increase (in some cases even $\Omega(D)$) 
of the treewidth of the constraint graph.

\subsection{New Results}
Our main result is summarized as the following theorem.
\begin{theorem}\label{thm:main}
For every instance $Q=(V,\DD,\HH,\CC)$ of CSP, there exists an
extended formulation $P(Q)$ of $CSP(Q)$ and $CSP'(Q)$
of size $\Oh(D^{\tau}n)$ where $\tau$ is the treewidth of
$Q$; moreover, the corresponding LP can be constructed 
in time $\Oh(D^{\tau}n)$.
\end{theorem}

As a corollary we obtain upper bounds on the extension complexity
for several NP-hard problems on the class of graphs with bounded
treewidth; as far as we know, these results have not been known.

\section{CSP Polytope}\label{sec:polytope}

\subsection{Integer Linear Programming Formulation}

We start by introducing the terms and notation that we use throughout
this section. We assume that $Q=(V,\DD,\HH,\CC)$ is a given instance 
of CSP.
For every
subset $W \subseteq V$ we define the set of all {\em configurations
  of} $W$ as
$$\CK(W) = \{(\alpha_1, \dots, \alpha_n)\ |\ 
 \forall C_U\in \HH \ (U\subseteq W \rightarrow \alpha|_{U}\in C_U) \mbox{ and }
\forall i \not\in U \ \alpha_i = \lambda\}\ $$
where $\lambda$ is a symbol not appearing in any of the domains $D_u$, $u\in V$.
For a configuration $K\in \CK(U)$ and $v\in V$, we use the
notation $K(v)$ to refer to the $v$-th element of $K$.
Also, for a configuration $K\in \CK(U)$, $v\in V\setminus U$
and $\alpha\in D_v$, we use the notation 
$K[v \leftarrow \alpha]$
to denote the configuration $K'\in \KK(U\cup\{v\})$ 
such that $K'(v) = \alpha$ and
$K'(u) = K(u)$ for every $u \neq v$.  

For an $n$-dimensional vector
$K=(\alpha_1, \dots, \alpha_n)$ and a subset of
variables $U \subseteq V$ we denote by $K\rest_U$ the
\textit{restriction of $K$ to $U$} that is defined as an $n$-dimensional
vector with
$K\rest_U(i) = K(i)$ for $i \in U$ and $K\rest_U(i) = \lambda$ for $i \not\in
U$ (i.e., we set to $\lambda$ all coordinates of $K$ outside of $U$). 
We denote by $\Lambda$ the configuration $(\lambda,
\dots, \lambda)\in \KK(\emptyset)$; 
note that for $\alpha\in D_v$, $\Lambda[v \leftarrow \alpha]$ is the
configuration from $\KK(\{v\})$ 
with exactly one non-$\lambda$ element, namely the $v$-th
element, equaling $\alpha$.

In our linear program, 
for every index $v \in V$ and every $i\in D_v$, 
we introduce a binary variable $y_v^i$. 
The task of the variable $y_v^i$ is to encode 
the value of the CSP-variable $z_v$: the variable $y_v^i$ is set to one if and
only if $z_v=i$. 
Since in every solution each variable assumes a unique value, we
enforce the constraint $\sum_{i\in D(v)} y_{v}^i = 1$ for each $v\in
V$.

For every configuration $K \in \bigcup_{U: C_U \in \CC\cup \HH} \KK(U)$ 
we introduce a binary variable $g(K)$. The intended meaning
of the variable $g(K)$, for $K\in \KK(U)$ and $U\subseteq V$, is to provide information 
about the values of the CSP-variables $z_u$ for $u\in U$ in the following 
way: $g(K)=1$ if and only if for every $u \in U$, $z_u = K(u)$.
To ensure consistency between the $y$ and $g$ variables, 
for every $C_U\in \CC\cup \HH$ and for every $v\in U$, 
we enforce the constraint $\sum_{K \in \KK(U): K(v)=i} g(K) = y_v^i$.
Note that for binary CSP, the $g$ variables capture the values
of CSP-variables $z$ for pairs of elements from $V$ that correspond
to edges of the constraint graph.

Relaxing the integrality constraints we obtain the following initial LP
relaxation of the CSP problem~$Q=(V,\DD, \HH, \CC)$:

\begin{align}
\sum_{i\in D(v)} y_{v}^i & = 1  & & \forall v\in V \label{LP2:sumtoone} \\
\sum_{K \in \KK(U): K(v)=i} g(K) & = y_v^i & & \forall C_U \in \CC\cup\HH \ \forall v\in U \ \forall i\in
D(v) \label{LP2:y_i}  \\
0 \leq \yy,\gg & \leq 1 & & 	\label{LP2:non-negative1}
\end{align}

Note that there is a one to one correspondence between the (extended)
feasible assignments of $Q$ and integral solutions of 
(\ref{LP2:sumtoone}) - (\ref{LP2:non-negative1}); from now on we denote
by $\textrm{proj}_1$ the linear projection of the convex hull of integral solutions of 
(\ref{LP2:sumtoone}) - (\ref{LP2:non-negative1}) to $CSP(Q)$.  
Also observe that the total weight of CSP-constraints satisfied by an integral 
vector $(\yy,\gg)$ satisfying (\ref{LP2:sumtoone}) - (\ref{LP2:non-negative1}) 
is\footnote{In the case of general payoff functions, the total weight is given by
$\sum_{C_U \in \CC} \sum_{K\in \KK(U):K|_U \in C_U} w(K|_U) g(K) $} 
\begin{align}
\sum_{C_U \in \CC}  w(C_U) \sum_{K\in \KK(U):K|_U \in C_U} g(K) \ . \nonumber
\end{align}

Unfortunately, even for CSP problems whose constraint graph is
series-parallel, the polytope given by the LP~(\ref{LP2:sumtoone}) -
(\ref{LP2:non-negative1}) is not integral
(consider, e.g., the instance of CSP corresponding to the independent 
set problem on $K_3$). The weakness of the formulation is
that no {\em global} consistency among the $\yy$ variables is guaranteed. To
strengthen the relaxation, we introduce new variables and constraints derived
from a tree decomposition of the constraint graph of $Q$.

\subsection{Extended Formulation}

Here we describe, for every CSP instance $Q=(V,\DD,\HH,\CC)$, a polytope $P(Q)$, 
and in the next subsection we
prove that $P(Q)$ is an extended formulation of $CSP(Q)$ and $CSP'(Q)$. 
The set of variables in the given LP description 
of $P(Q)$ is substantially different from the set of variables 
used in the LP~(\ref{LP2:sumtoone}) - (\ref{LP2:non-negative1}), and the set
of new constraints is completely different from the the set of constraints
in the LP~(\ref{LP2:sumtoone})~-~(\ref{LP2:non-negative1}).
Whereas in the previous subsection, there is (roughly) a variable $g(K)$ for 
every feasible assignment of every subset of CSP variables corresponding
to a {\em soft or hard constraint}, here we have a variable for every
feasible assignment of every subset of CSP variables corresponding
to a {\em bag in a given tree decomposition of the constraint graph}.
Nevertheless, as we show after defining $P(Q)$,
there exists a simple linear projection of $P(Q)$ to the convex hull of
all integral points in the polytope given by the 
LP~(\ref{LP2:sumtoone})~-~(\ref{LP2:non-negative1}). 

Let $T=(V_T,E_T)$ be a fixed nice tree decomposition~\cite{Kloks:94} 
of the constraint graph of
$Q$ and for every node $a\in V_T$, let $B(a)\subseteq V$ denote the
corresponding bag. 
Let $\BB = \{B(a) \ |\ a \in V_T\}$ denote the set of all bags of $T$.  
Let $\KK_\BB=\bigcup_{B\in \BB} \KK(B)$ be the set of all configurations
of all bags in $T$.
We use $V_I\subseteq V_T$ to
denote the subset of all introduce nodes in $T$ and $V_F \subseteq V_T$ to denote
the subset of all forget nodes in $T$.

For every configuration $K\in \KK_\BB$ we introduce a binary variable $f(K)$.
As in the previous subsection, the intended meaning
of the variable $K\in \KK(B)$, for $B\in \BB$, is to provide information
about the values of the CSP-variables $z_u$ for $u\in B$ in the following
way: $f(K)=1$ if and only if for every $u \in B$, $z_u = K(u)$.
To ensure consistency among variables indexed by the configurations of
the same bag, namely to ensure that for every $B \in \BB$ there 
exists exactly one configuration $K \in \KK(B)$ with $f(K) = 1$, we
introduce for every $B \in \BB$ the LP constraint 
$\sum_{K \in \KK(B)} f(K) = 1$.

For every introduce node $c\in V_T$ with a child $b\in V_T$ and for
every configuration $K\in \KK(B(b))$ we have the constraint
$\sum_{K' \in \KK(B(c)): K'\ \rest_{B(b)} = K}f(K') = f(K)$, 
and symmetrically,
for every forget node $c\in V_T$ with a child $b\in V_T$ and for every
configuration $K\in \KK(B(c))$ we have the constraint
$\sum_{K' \in \KK(B(b)): K'\ \rest_{B(c)} = K}f(K') = f(K)$.

Relaxing the integrality constraints and putting all these additional 
constraints together, we obtain:
\begin{align}
\sum_{K\in \KK(B)} f(K)  & = 1 & & \forall B\in
                                   \BB \label{LP:configurations} \\
\sum_{K' \in \KK(B(c)): K'\ \rest_{B(b)} = K}f(K') & =  f(K) & &
\forall c\in V_I, \forall K\in \KK(B(b)) \mbox{ where $b$ is } \label{LP:introduce}  \\[-14pt]
& & & 	\mbox{the only child of $c$}	\nonumber \\[3pt]
\sum_{K' \in \KK(B(b)): K'\ \rest_{B(c)} = K}f(K') & =  f(K) & &
\forall c\in V_F, \forall K\in \KK(B(c)) \mbox{ where $b$ is } \label{LP:forget} \\[-14pt]
& & & 	\mbox{the only child of $c$}	\nonumber \\[3pt]
0\leq \ff & \leq 1  & &  \label{LP:f_nonnegative}
\end{align}
For the given binary CSP instance $Q$, 
we denote the polytope associated with the LP (\ref{LP:configurations}) - (\ref{LP:f_nonnegative}), as $P(Q)$.

Consider now a vector $\ff\in P(Q)$ and the following set of linear equations:
\begin{align}
y_{v}^i & = \sum_{{K\in \KK(B) : K(v)=i}} f(K)  & & \forall B\in\BB,  \forall v\in B, \forall i\in D_v  \\
g(K) & = \sum_{K'\in \KK(B) : K'\ \rest_U = K} f(K') & & \forall B\in\BB, 
	\forall C_U \in \CC\cup\HH \mbox{ s.t. } U\subseteq B, \forall K\in \BB(U) 
\end{align}

It is just a technical exercise to check that for a given $\ff\in P(Q)$, 
there always exists a unique solution $(\yy, \gg)$ of this LP and that
the unique $(\yy, \gg)$ is a linear projection of $\ff$. Moreover,
such a vector $(\yy, \gg)$ also satisfies the LP 
constraints~(\ref{LP2:sumtoone}) - (\ref{LP2:non-negative1}).
The point is that there exists a linear projection of $P(Q)$ into the polytope
defined by the LP~(\ref{LP2:sumtoone}) - (\ref{LP2:non-negative1});
moreover, an integral point from $P(Q)$ is mapped on an integral point.
From now on we denote this projection $\textrm{proj}_2$.

\subsection{Proof of Theorem~\ref{thm:main}}\label{subsec:proof}

As in the previous
subsections, we assume that $Q=(V,\DD,\HH,\CC)$ is a given instance of
CSP, $G=(V,E)$ is the constraint graph of $Q$ and $T=(V_T, E_T)$ a
fixed nice tree decomposition of $G$.  
We start by introducing several notions that will help us dealing with
tree decompositions and our linear program.  

For a node $a\in V_T$, let
$T(a)=(V_{a},E_{a})$ be the subtree of $T$ rooted in $a$;
the {\em configurations relevant to} $T(a)$ are those in the set 
$\RP(a)=\bigcup_{b \in V_{a}} \KK(B(b))$, and the {\em
variables relevant to} $T(a)$ are those $f(K)$ for which
$K\in \RP(a)$. 
For succinctness of notation, we denote the projection $\ff|_{\RP(a)}$ 
of the vector $\ff$ on the set of variables relevant to $T(a)$ 
also by $\ff|_{a}$.
The {\em constraints relevant to} $T(a)$ are those containing only the
variables relevant to $T(a)$.
We say that a vector $I\in \{0,1\}^{\RP(a)}$  {\em agrees with the
configuration} $K\in \RP(a)$ if $I(K)=1$.

Let $\ff$ be a fixed solution of the LP~(\ref{LP:configurations}) -
(\ref{LP:f_nonnegative}) that corresponds to a vertex of the polytope 
$P(Q)$. Our main tool is the following lemma.

\begin{lemma}\label{lem:combination}
  For every 
  node $b\in V_T$,
  there exist a positive integer $M$ and binary vectors 
  $I_1, I_2, \dots, I_M \in\{0,1\}^{\RP(b)}$, 
  some possibly identical, such that
\begin{itemize} 
\item[$\spadesuit$] every $I_i$ satisfies the constraints relevant to
  $T(b)$,
\item[$\clubsuit$] $\ff|_{b} = \frac{1}{M}\sum_{i=1}^M  I_i  $.
\end{itemize}
\end{lemma}

\begin{proof}
By induction. We start in the leaves of $T$ 
and proceed in a bottom-up fashion.

\paragraph{Base case.}
Assume that $b\in V_T$ is a leaf of the nice decomposition
tree $T$. By definition of a
nice tree decomposition, the bag $B(b)$ consists of a single vertex,
say a vertex $v\in V$.  The only variables relevant to $T(b)$ are $f(K)$ for all 
$K\in \KK(B(b))=\bigcup_{j\in D(v)}\Lambda[v \leftarrow j]$, and the only
relevant constraints are those of the type (\ref{LP:configurations}) and
(\ref{LP:f_nonnegative}). 

Let $M' \in \mathbb{N}$ be such that an $M'$-multiple of every relevant
variable is integral; as $\ff$ is a solution corresponding 
to a vertex of the polytope $P(Q)$, all the variables are rational
which guarantees that such an $M'$ exists.  
For every $j\in D_v$ we define an integral vector
$I_j$ such that
$I_j(\Lambda[v \leftarrow j])=1$
and $I_j(\Lambda[v \leftarrow i])=0$ for every $i\neq j$.

The vector $I_j$ will appear with multiplicity $M'\cdot y_v^j$ among
the integral solutions $I_1,\ldots,I_{M'}$ for $G'$.  
Then, obviously, both properties $\spadesuit$ and $\clubsuit$ are satisfied.

\paragraph{Inductive step.}
Consider an internal node $c\in V_T$ of the nice decomposition tree~$T$.  We
distinguish three cases: $c$ is a join node, $c$ is an
introduce node and $c$ is a forget node.

{\em Join node.}  Assume that the two children of the join node $c$ are
$a$ and $b$. Recall that $B(a)=B(b)=B(c)$.  By the inductive
assumption, there exist integers $M$ and $M'$ and integral vectors
$I_1,\ldots,I_M\in\{0,1\}^{\RP(a)}$, 
each of them satisfying the relevant constraints for
$T(a)$ and such that
$\ff|_{a} = \frac{1}{M}\sum_{i=1}^M I_i $, and integral
vectors $J_1,\ldots,J_{M'}\in\{0,1\}^{\RP(b)}$, 
each of them satisfying the relevant
constraints for $T(b)$ and such that
$\ff|_{b} = \frac{1}{M'}\sum_{i=1}^{M'} J_i $.  

Two vectors $I_i$ and $J_j$ that agree with a given configuration
$K \in \KK(B(c))$ can be easily merged into an integral vector
$L\in\{0,1\}^{\RP(c)}$ 
that satisfies $L|_{a}=I_i$ and $L|_{b}=J_j$; as the set of all 
constraints relevant to $T(c)$ is the union of the constraints relevant to 
$T(a)$ and the constraints relevant to $T(b)$, the vector
$L$ satisfies also all the constraints relevant to $T(c)$.

For simplicity we assume, without loss of generality, that $M=M'$.  
Then, by the property $\clubsuit$ and since
$B(a) = B(b) = B(c)$, for every configuration $K\in \KK(B(c))$, the
number of vectors $I_i$ that agree with $K$ is equal to the number of
vectors $J_j$ that agree with $K$, namely $M\cdot f(K)$. Thus, it is
possible to match the vectors $I_i$ and $J_j$ one to one in such a way
that both vectors in each pair agree with the same configuration; let
$L_1, L_2, \ldots, L_M$ denote the result of their merging as
described above.  Then the vectors $L_i$ satisfy the property
$\spadesuit$ as explained in the previous paragraph, and by
construction they also satisfy the property $\clubsuit$.

{\em Introduce node.}  Assume that the only child of the introduce
node $c$ is a node $b$ and $B(c)=B(b)\cup \{v\}$. By the inductive
assumption, there exists integer $M$ and integral vectors
$I_1,\ldots,I_M\in\{0,1\}^{\RP(b)}$, each of them satisfying the
relevant constraints for $T(b)$ and such that
$\ff|_{b} = \frac{1}{M}\sum_{i=1}^M I_i$.  Without loss of
generality we assume that for every variable relevant to $T(c)$, its
$M$-multiple is integral. We partition the vectors $I_1,\ldots,I_M$
into several groups indexed by the configurations from $\KK(B(b))$:
the group $Z_K$, for $K\in \KK(B(b))$, consists exactly of those
vectors $I_i$ that agree with $K$.

Consider a fixed configuration $K\in \KK(B(b))$ and the corresponding
group $Z_K$.  Note that the size of this group is $M\cdot f(K)$. 
We further partition the group $Z_K$ into at most
$|D_v|$ subgroups $Z_{K'}$, where $K' = K[v \leftarrow j]$, for every
$j\in D_v$ satisfying $K[v \leftarrow j]\in \KK(B(c))$, 
in such a way that $Z_{K'}$ contains exactly
$M\cdot f(K')$ vectors (it does not matter which ones); the LP
constraint~(\ref{LP:introduce}) makes this possible. Then, for every
$j\in D_v$, we create from every vector $I\in Z_{K[v \leftarrow j]}$ a new integral
vector $J_{I}$ in the following way: 
\begin{itemize}
\item for every $\bar K\in \RP(b)$, $J_I(\bar K)=I(\bar K)$; this guarantees 
		$J_{I}|_{b}=I$,
\item $J_I(K[v \leftarrow j])=1$,
\item for every $i\in D_v$, $i\not =j$, $J_I(K[v \leftarrow i])=0$. 
\end{itemize}

Obviously, the new vectors $J_{I}$
satisfy all constraints relevant to
$T(b)$, and it is easy to check that they satisfy all constraints
relevant to $T(c)$ as well, given the definitions above.  Moreover, 
the definitions
above imply that the vectors $J_{I}$ satisfy the property $\clubsuit$.

{\em Forget node.}  Assume that the only child of the forget node $c$
is a node $b$, $B(c)=B(b)\setminus \{v\}$. This case is symmetric to
the previous one in that instead of splitting the groups $Z_K$ into
smaller groups $Z_{K'}$, we merge them into bigger $Z_{K'}$. 

By the inductive assumption, there exists an integer $M$
and integral vectors $I_1,\ldots,I_M\in\{0,1\}^{\RP(b)}$, 
each of them satisfying the
relevant constraints for $T(b)$ and such that
$\ff|_{b} = \frac{1}{M}\sum_{i=1}^M I_i$.  Without loss
of generality we assume that for every variable relevant to $T(c)$, its
$M$-multiple is integral. We partition the vectors $I_1,\ldots,I_M$
into several groups indexed by the configurations from $\KK(B(b))$: the
group $Z_K$, for $K\in \KK(B(b))$, consists exactly of those vectors
$I_i$ that agree with $K$. Note that the size of $Z_K$ is $M\cdot f(K)$. 

For every $K' \in \KK(B(c))$ we create a bigger group group $Z_{K'}$
by merging $|D_v|$ of the groups $Z_K$, namely those satisfying
$K|_{B(c)} = K'$.
By the LP constraint~(\ref{LP:forget}), the new group $Z_{K'}$ contains
exactly $M\cdot f(K')$ vectors. 
For every $K' \in \KK(B(c))$, we create from every vector $I\in Z_{K'}$ a
new integral vector $J_{I}$ in the following way:
\begin{itemize}
\item for every $\bar K\in\RP(b)$, $J_I(\bar K)=I(\bar K)$.
\end{itemize}
If $\KK(B(c))\subseteq \RP(b)$, there is nothing more to do. Otherwise
we further define
\begin{itemize}
\item $J_I(K')=1$, and
for every $\hat{K} \in \KK(B(c))$, $\hat{K}\not = K'$,
$J_I(\hat{K}) = 0$.
\end{itemize}

We have to check that the vectors $J_{I}$ satisfy all constraints
relevant to $T(c)$. The only possibly 
new constraints are those using variables $f(K')$ for
$K' \in \KK(B(c))$ and it is easily seen that they are satisfied,
given the definitions above. Also, the definitions above imply that
the vectors $J_{K'}$ satisfy the property $\clubsuit$.
\qed
\end{proof}

By applying Lemma~\ref{lem:combination} 
to the whole tree $T$, that is, to the subtree
rooted in the root of $T$, we immediately obtain that $\ff$ is an
integral vector, and, thus, also the corresponding vertex of $P(Q)$ 
is integral. As this holds for every vertex of $P(Q)$, we conclude that
$P(Q)$ is an integral polytope.

Considering the notes at the ends of the previous two subsections,
we also conclude that $CSP(Q) = \textrm{proj}_1(\textrm{proj}_2(P(Q))$ and 
$CSP'(Q)=\textrm{proj}_V(CSP(Q))$.

To complete the proof of Theorem~\ref{thm:main}, we observe that
the number of variables and constraints in the
LP~(\ref{LP:configurations}) - (\ref{LP:f_nonnegative}) 
is $\Oh(D^{\tau}n)$. \qed

\section{Applications}\label{sec:problems}

The purpose of this section is to make explicit the extension complexity upper
bounds given in Theorem~\ref{thm:main} for several well known graph problems.
We find it interesting that the attained extension complexity upper bounds meet
the best possible (assuming Strong ETH)  \textit{time complexity} lower bounds,
given by Lokshtanov et al. \cite{LMS:11}; the only exception is the
\textsc{Multiway Cut} problem.
To state our results, we use for each problem
the following template:

\begingroup
\setlength{\parindent}{0em}
\hrulefill

\textlineii[t]{Problem name}{Projection}{Extension complexity}{Time complexity}

\textsc{Instance:} \ \ \ldots

\textsc{Solution:} \ \ \ldots

\textsc{CSP formulation:} $V$, $\DD$, $\HH$,
$\CC$. \hfill CSP version: Decision / \textsc{Max} / \textsc{Min}

where {\em Projection} is the name of the linear projection that yields the
natural polytope of the problem $Q$ from the $CSP(Q)$ polytope 
(or from the $P(Q)$ polytope, in case of the OCT problem). 
We use the notation $[n] = \{1, \dots, n\}$. 

\hrulefill


\problemline{Coloring / Chromatic Number
  \cite{ACGKMP:99}}{V}{$q^\tau n$}{$q^\tau n$}

\textsc{Instance:} Graph $G=(V,E)$, set of colors $[q]$

\textsc{Solution:} A coloring of $G$ with $q$ colors with no
monochromatic edges.

\textsc{CSP formulation:} $V = [n]$, $D_v = [q]$ for every $v \in V$,
$H_{uv} = \{(i,j) \ | \ i \in D_u, j \in D_v, i \neq j\}$ for every $uv \in
E$, $\CC = \emptyset$. \hfill Decision

\textbf{Comment:} Note that \textsc{Chromatic Number} $\chi(G)$ of $G$ 
is always upper bounded by $\tau+1$ since graphs of bounded treewidth 
are $\tau$-degenerate and thus $(\tau+1)$-colorable. Thus, if the goal 
is to determine $\chi(G)$, it suffice to find the smallest $q$ 
such that $CSP(Q)$ is non-empty.

\hrulefill
  \problemline{List-$H$-Coloring / List
    Homomorphism \cite{FH:98}}{V}{$L^\tau n$}{$L^\tau n$}

\textsc{Instance:} Graph $G=(V,E)$, graph $H=(V_H, E_H)$
  possibly containing loops, and for every vertex
  $v \in V$ a set $L(v) \subseteq V_H$. (We
  denote $L = \max_{v \in V} |L(v)|$)

\textsc{Solution:} A mapping $f: V \rightarrow V_H$ such that $\forall
uv \in E$ it holds that $f(u)f(v) \in E_H$ and $f(v) \in L(v)$ for every
$v \in V$.

\textsc{CSP formulation:} $V = [n]$, $D_v = L(v)$ for every $v \in V$,
$H_{uv} = \{(i,j) \ | \ i \in D_u, j \in D_v, ij \in E_H\}$ for every
$uv \in E$, $\CC = \emptyset$.\hfill Decision

\textbf{Comment:} Note that the problems \textsc{List
  Coloring}, \textsc{Precoloring Extension} and \textsc{$H$-Coloring}
(or \textsc{Graph Homomorphism}) are all special cases of this
problem. The lower bound given by Lokshtanov et al. \cite{LMS:11}
applies to all of them since \textsc{Coloring} is a special case
of each of them.

\hrulefill


\problemlinealt{Unique Games \cite{Khot:02}}{\proj{id}}{$\Oh(t^\tau n)$}{---}

\textsc{Instance:} Graph $G=(V,E)$, an integer $t \in \NN$, a
permutation $\pi_e$ of order $t$ for every edge $e \in E$.

\textsc{Solution:} A mapping $\ell: V \rightarrow [t]$ such that the
number of edges $uv \in E$ with $\pi_{uv}(\ell(u)) = \ell(v)$ is maximized.

\textsc{CSP formulation:} $V = [n]$, $D_v = [t]$ for every $v \in V$, $\HH =
\emptyset$, $C_{uv} = \{(i,\pi_{uv}(i)) \ | \ i \in D_u \}$ for every edge $uv \in E$.\hfill \textsc{Max}

\textbf{Comment:} The decision variant of this problem is not
interesting as it is trivially solvable in polynomial time.

\hrulefill


\problemlinealt{Multiway Cut
  \cite{ACGKMP:99}}{\proj{E}}{$\Oh(t^\tau n)$}{$\Oh(t^\tau n)$}

\textsc{Instance:} Graph $G=(V,E)$, an integer $t \in \NN$ and $t$ vertices 
$s_1, \dots, s_t \in V$

\textsc{Solution:} A partition of $V$ into sets $V_1, \dots, V_t$ such
  that for every $i$ we have $s_i \in V_i$ and the total number of edges 
  between $V_i$ and $V_j$ for $i \neq j$ is minimized.

\textsc{CSP formulation:} $V=[n]$, $D_v = [t]$ for every $v \in V$, $\HH =
\emptyset$, $C_{uv} = \{(i,i) \ | \ i \in [n]\}$ for every edge
$uv \in E$.\hfill \textsc{Min}

\textbf{Comment:} Setting $z_v = i$ models vertex $v$ belonging to the set
$V_i$. Not satisfying the constraint $C_{uv}$ means that the edge $uv$ 
belongs to the multiway cut. 

\hrulefill


  \problemline{Max Cut \cite{ACGKMP:99}}{E}{$2^\tau n$}{$2^\tau n$}

\textsc{Instance:} Graph $G=(V,E)$

\textsc{Solution:} A partition of vertices into two sets $V_1, V_2$ such
that the number of edges between $V_1$ and $V_2$ is maximized.

\textsc{CSP formulation:} $V=[n]$, $D_v = \{0,1\}$ for every $v \in V$,
$\HH = \emptyset$, $C_{uv} = \{(1,0), (0,1)\}$ for every edge $uv \in
E$.\hfill \textsc{Max}

\textbf{Comment:} The values $0,1$ model the vertex belonging to
the set $V_1$ or $V_2$. If we replace maximization by minimization,
the problem
becomes \textsc{Edge Bipartization} (aka \textsc{Edge OCT})
problem which is a parametric dual of \textsc{Max~Cut}.

\hrulefill

  \problemline{Vertex Cover \cite{ACGKMP:99}}{V}{$2^\tau n$}{$2^\tau n$}

\textsc{Instance:} Graph $G=(V,E)$

\textsc{Solution:} A set of vertices $C \subseteq V$ of minimal size such
that every edge contains a vertex $v \in C$ as at least one of its
endpoints.

\textsc{CSP formulation:} $V=[n]$, $D_v = \{0,1\}$ for every $v \in V$,
$H_{uv} = \{(0,0), (0,1), (1,0)\}$ for every edge $uv \in
E$, $C_v = \{1\}$.\hfill \textsc{Min}

\textbf{Comment:} The values $0,1$ model the vertex belonging to
$C$ or $V \setminus C$. If we replace maximization by minimization, 
the problem becomes \textsc{Independent Set} problem which is a parametric dual of
\textsc{Vertex~Cover}.

\hrulefill

\problemlinealtt{Odd Cycle Transversal \cite{LMS:11}}{\proj{OCT} $\circ$
  \proj{2}}{$\Oh(3^\tau
  n)$}{$\Theta(3^\tau n)$}

\textsc{Instance:} Graph $G=(V,E)$

\textsc{Solution:} A subset of vertices $W \subseteq V$ of minimal size 
such that $G[V \setminus W]$ is a bipartite graph.

\textsc{CSP formulation:} $V=[n]$, $D_v = \{0,1,2\}$ for every $v \in V$,
$H_{uv} = \{0,1,2\}^2 \setminus \{(0,0), (1,1)\}$ for every edge $uv \in
E$, $C_v = \{0,1\}$ for every $v \in V$.\hfill \textsc{Min}

\textbf{Comment:} The values $0,1,2$ model the vertex belonging
to either the first or the second partite of a bipartite graph, or the
deletion set $W$. Satisfying the constraint $C_v$ corresponds to not
putting $v$ in the deletion set $W$. Also known as \textsc{Vertex Bipartization}.
The projection $\textrm{proj}_{OCT}:P(Q)\rightarrow  \{0,1\}^V$ is defined as follows: 
 $\textrm{proj}_{OCT}(y_1^0,y_1^1,y_1^2,y_2^0,y_2^1,y_2^2,\ldots,
y_n^0,y_n^1,y_n^2,\gg)=(y_1^2,y_2^2,\ldots,y_n^2)$.


\endgroup

\section{Open problems} 
A natural research direction is to examine more closely the extension
complexity for CSP and the specific graph problems on graphs with
bounded treewidth, in particular, what are the best possible upper bounds? 
\vspace{-2mm}
\paragraph{Acknowledgments.} The  authors thank Hans Raj Tiwary and
Ji\v r\'{\i} Sgall for stimulating discussions.

\bibliographystyle{abbrv}
\bibliography{lcuts}

\end{document}